\documentclass[11pt]{article}

\usepackage[a4paper,margin=30mm]{geometry}
\usepackage{amsmath,amssymb,amsthm,mathtools}
\usepackage{microtype}
\usepackage[hidelinks]{hyperref}

\newtheorem{theorem}{Theorem}[section]

\newcommand{\R}{\mathbb{R}}
\newcommand{\C}{\mathbb{C}}
\newcommand{\F}{\mathcal{F}}
\newcommand{\dd}{\,\mathrm{d}}
\DeclareMathOperator{\Tr}{Tr}
\DeclareMathOperator{\Str}{Str}

\title{Finite Rodr\'iguez-Villegas Approximants\\
to the Riemann $\xi$-Function}
\author{%
  Alisson M. Silva\\[3pt]
  \small Instituto de F\'isica, Universidade de S\~ao Paulo\\
  \small Rua do Mat\~ao, 1371, Cidade Universit\'aria\\
  \small 05508-090 S\~ao Paulo, SP, Brazil\\
  \small \href{mailto:alisson.matheus.silva@usp.br}{alisson.matheus.silva@usp.br}
}
\date{}

\begin{document}
\maketitle

\begin{abstract}
We construct a sequence of finite Rodr\'iguez-Villegas transforms converging
locally uniformly to the Riemann $\xi$-function in the critical strip.  The
input is a positive symmetric profile on the unit interval obtained from the
Riemann theta kernel through convolution with the hyperbolic-secant kernel and
the logistic coordinate.  The profile is a Stieltjes function of $x(1-x)$.
Its Bernstein polynomials produce reciprocal numerators and exact finite
functional equations.  The same numerators admit an exact realization as
fermionic supertraces, while the Bernstein polynomials are normalized Gibbs
traces.
\end{abstract}


\section{Introduction}

Rodr\'iguez-Villegas considered a rational function
\begin{equation}
  \frac{U(q)}{(1-q)^d}=\sum_{n\geq0}h_nq^n
  \label{eq:rv-generating-function}
\end{equation}
and the polynomial $H$ for which $h_n=H(n)$ for all sufficiently large $n$.
In commutative algebra, \eqref{eq:rv-generating-function} is the shape of a
Hilbert--Poincar\'e series, while $H$ is its eventual Hilbert polynomial.  If
$U(1)\ne0$ and every zero of $U$ lies on the unit circle,
Rodr\'iguez-Villegas proved that the nontrivial zeros of $H$ lie on a vertical
line \cite{RodriguezVillegas2002}.  In the special case $\deg U=d-1$ and
$U\in\R[q]$, the polynomial $z(s)=H(-s)$ satisfies
$z(1-s)=(-1)^{d-1}z(s)$, and all of its zeros lie on $\Re s=1/2$.  One of the
motivating examples was the Ehrhart polynomial of the $n$-dimensional
cross-polytope.  Bump, Choi, Kurlberg, and Vaaler had established its
functional equation and vertical-line zero theorem by orthogonality
\cite{BumpChoiKurlbergVaaler2000}.  The Rodr\'iguez-Villegas theorem recovers
the same phenomenon from the numerator $(1+q)^n$.  The final section of
\cite{RodriguezVillegas2002} proposed a possible
infinite-dimensional interpretation behind this finite phenomenon.  In
concrete terms, it asked whether $\zeta(-s)$ could be understood as the Hilbert
function of a natural infinite-dimensional, possibly Gorenstein, graded
algebra, and what its Poincar\'e series should be.  This proposal did not
construct such an algebra or an associated limiting procedure.

Subsequent work supplied arithmetic instances of this finite mechanism.
Conrey, Farmer, and {\"O}.~{\.{I}}mamo\u{g}lu proved a unit-circle theorem for
odd period polynomials \cite{ConreyFarmerImamoglu2013}, and the corresponding
theorem for full period polynomials of newforms was established by Jin, Ma,
Ono, and Soundararajan \cite{JinMaOnoSoundararajan2016}.  Manin applied the
Rodr\'iguez-Villegas transform to period polynomials of Hecke eigenforms and
obtained polynomials with a zeta-type functional equation and critical-line
zeros \cite{Manin2016}.  Ono, Rolen, and Sprung then constructed
zeta-polynomials encoding critical values of modular $L$-functions
\cite{OnoRolenSprung2017}.  Jameson clarified their Hilbert-polynomial and
Eichler--Shimura structure \cite{Jameson2019}, while L\"obrich, Ma, and
Thorner developed motivic analogues for symmetric-power $L$-functions
\cite{LobrichMaThorner2017}.  This literature realizes the finite
unit-circle-to-critical-line principle with substantial arithmetic content.

The analytic target in the present note is the completed Riemann zeta
function
\begin{equation}
  \xi(s)=\frac12s(s-1)\pi^{-s/2}\Gamma(s/2)\zeta(s),
  \qquad
  \Xi(t)=\xi\!\left(\frac12+it\right).
  \label{eq:xi-definition}
\end{equation}
We approach the Rodr\'iguez-Villegas proposal in the reverse direction:
starting from the theta representation of $\xi$, we extract a positive
symmetric profile $P_\xi$ on $[0,1]$ and approximate it by Bernstein
polynomials.  Their samples define reciprocal finite numerators.  The
associated Rodr\'iguez-Villegas transforms satisfy exact functional equations
and converge locally uniformly to $\xi$ in the critical strip.  The same
finite numerators arise as supertraces on exterior algebras, giving the
approximation a canonical fermionic formulation.  In this way the original
Hilbert-series mechanism is retained while its input is reconstructed from
the analytic theta kernel.

\section{The theta--logistic profile}

Use the standard positive even Riemann theta kernel $\Phi$, normalized by
\begin{equation}
  \Xi(t)=\int_{\R}\Phi(u)e^{itu}\dd u.
  \label{eq:riemann-fourier}
\end{equation}
The classical theta representation and its decay properties may be found, for
example, in \cite{Titchmarsh1986}.  Put
\begin{equation}
  k(u)=\frac{1}{2\cosh(u/2)},
  \qquad
  \widehat{k}(t)=\frac{\pi}{\cosh(\pi t)},
  \qquad
  w_\xi=\Phi*k.
  \label{eq:sech-kernel}
\end{equation}
For
\begin{equation}
  x=\frac{e^u}{1+e^u},
  \qquad
  u=\log\frac{x}{1-x},
  \label{eq:logistic-coordinate}
\end{equation}
define
\begin{equation}
  P_\xi(x)=
  \frac{w_\xi\!\left(\log(x/(1-x))\right)}{\sqrt{x(1-x)}}.
  \label{eq:profile-definition}
\end{equation}
We use the standard terminology for Stieltjes and completely monotone
functions as in \cite{SchillingSongVondracek2012}.

\begin{theorem}\label{thm:profile}
The function $P_\xi$ extends continuously to $[0,1]$, is positive and
symmetric, and satisfies
\begin{equation}
  \xi(s)=\frac{\sin(\pi s)}{\pi}
  \int_0^1 P_\xi(x)x^{s-1}(1-x)^{-s}\dd x,
  \qquad 0<\Re s<1.
  \label{eq:continuum-rv}
\end{equation}
It has the representation
\begin{equation}
  P_\xi(x)=\int_{\R}\Phi(v)
  \frac{e^{v/2}}{x+(1-x)e^v}\dd v.
  \label{eq:theta-profile}
\end{equation}
Moreover, writing $p=x(1-x)$,
\begin{equation}
  P_\xi(x)=G(p),
  \qquad
  G(p)=\int_0^\infty\frac{\rho(\dd v)}{1+p\lambda(v)},
  \label{eq:stieltjes-profile}
\end{equation}
where
\begin{equation}
  \rho(\dd v)=2\Phi(v)\cosh(v/2)\dd v,
  \qquad
  \lambda(v)=4\sinh^2(v/2).
  \label{eq:rho-lambda}
\end{equation}
Consequently, $G$ is completely monotone on $[0,\infty)$:
\begin{equation}
  (-1)^nG^{(n)}(p)
  =n!\int_0^\infty
  \frac{\lambda(v)^n\rho(\dd v)}{(1+p\lambda(v))^{n+1}}>0.
  \label{eq:complete-monotonicity}
\end{equation}
\end{theorem}

\begin{proof}
For $0<\Re s<1$, the substitution $y=e^u$ gives
\begin{equation}
  \int_{\R}k(u)e^{(s-1/2)u}\dd u
  =\int_0^\infty\frac{y^{s-1}}{1+y}\dd y
  =\frac{\pi}{\sin(\pi s)}.
  \label{eq:k-laplace}
\end{equation}
The theta kernel decays rapidly enough to justify convolution and Fubini in
this strip.  The bilateral theta representation of $\xi$ and
\eqref{eq:k-laplace} therefore imply
\begin{equation}
  \int_{\R}w_\xi(u)e^{(s-1/2)u}\dd u
  =\frac{\pi}{\sin(\pi s)}\xi(s).
\end{equation}
Since $k(u)=\sqrt{x(1-x)}$ and
$\dd u=\dd x/[x(1-x)]$, the logistic substitution yields
\eqref{eq:continuum-rv}.

Next, write $w_\xi(u)=\int_{\R}\Phi(v)k(u-v)\dd v$ and divide by
$k(u)$.  Direct simplification in the variable $x=e^u/(1+e^u)$ gives
\eqref{eq:theta-profile}.  Pairing the contributions of $v$ and $-v$ and
using the evenness of $\Phi$ gives
\begin{equation}
  P_\xi(x)=\int_0^\infty
  \frac{2\Phi(v)\cosh(v/2)}
       {1+4x(1-x)\sinh^2(v/2)}\dd v,
\end{equation}
which is \eqref{eq:stieltjes-profile}.  Positivity, symmetry, continuity at
the endpoints, and \eqref{eq:complete-monotonicity} follow from this formula
and differentiation under the integral sign.
\end{proof}

\section{Finite Rodr\'iguez-Villegas approximants}

For an even positive integer $E$, let
\begin{equation}
  B_EP_\xi(x)=\sum_{j=0}^E
  P_\xi(j/E)\binom{E}{j}x^j(1-x)^{E-j}
  \label{eq:bernstein-polynomial}
\end{equation}
be the Bernstein polynomial of $P_\xi$, in the notation of
\cite{Lorentz1986}.  Define
\begin{equation}
  U_E(q)=\sum_{j=0}^E u_{E,j}q^j,
  \qquad
  u_{E,j}=(-1)^j\binom{E}{j}P_\xi(j/E),
  \label{eq:numerator}
\end{equation}
and its finite Rodr\'iguez-Villegas transform by
\begin{equation}
  Z_E(s)=\sum_{j=0}^E u_{E,j}\binom{E-j-s}{E}.
  \label{eq:finite-rv-transform}
\end{equation}
The generalized binomial coefficient is understood as a polynomial in $s$.
Equivalently, $Z_E(s)=H_E(-s)$, where $H_E$ is the polynomial eventually
represented by the coefficients of
$U_E(q)/(1-q)^{E+1}$.

\begin{theorem}\label{thm:finite-approximants}
For every even $E$,
\begin{equation}
  U_E(q)=(1-q)^E
  B_EP_\xi\!\left(-\frac{q}{1-q}\right)
  \label{eq:mobius-bernstein}
\end{equation}
and
\begin{equation}
  U_E(q)=q^EU_E(q^{-1}).
  \label{eq:reciprocity}
\end{equation}
For $0<\Re s<1$ one has the exact integral formula
\begin{equation}
  Z_E(s)=\frac{\sin(\pi s)}{\pi}
  \int_0^1 B_EP_\xi(x)x^{s-1}(1-x)^{-s}\dd x.
  \label{eq:finite-integral}
\end{equation}
Consequently,
\begin{equation}
  Z_E(1-s)=Z_E(s)
  \label{eq:finite-functional-equation}
\end{equation}
and
\begin{equation}
  Z_E(s)\longrightarrow\xi(s)
  \label{eq:local-uniform-limit}
\end{equation}
locally uniformly in the strip $0<\Re s<1$.
\end{theorem}

\begin{proof}
The M\"obius identity \eqref{eq:mobius-bernstein} follows by substituting
$x=-q/(1-q)$ in \eqref{eq:bernstein-polynomial}.  Since
$P_\xi(j/E)=P_\xi(1-j/E)$ and $E$ is even, the coefficients in
\eqref{eq:numerator} satisfy $u_{E,j}=u_{E,E-j}$, proving
\eqref{eq:reciprocity}.

Expanding the Bernstein polynomial and applying the beta integral gives
\begin{align}
  &\frac{\sin(\pi s)}{\pi}
  \int_0^1 B_EP_\xi(x)x^{s-1}(1-x)^{-s}\dd x \notag\\
  &\quad=
  \sum_{j=0}^E\binom{E}{j}P_\xi(j/E)
  \frac{\sin(\pi s)}{\pi}
  \frac{\Gamma(s+j)\Gamma(E-j+1-s)}{\Gamma(E+1)}.
  \label{eq:beta-expansion}
\end{align}
The reflection formula transforms the $j$th summand into
$u_{E,j}\binom{E-j-s}{E}$, which proves \eqref{eq:finite-integral}.
The symmetry $B_EP_\xi(1-x)=B_EP_\xi(x)$ and the substitution
$x\mapsto1-x$ then yield \eqref{eq:finite-functional-equation}.

Finally, the Bernstein theorem gives
$\|B_EP_\xi-P_\xi\|_\infty\to0$.  If $K$ is compact in the critical strip,
there is a $\delta>0$ such that
$\delta\le\Re s\le1-\delta$ on $K$.  Equations
\eqref{eq:continuum-rv} and \eqref{eq:finite-integral} imply
\begin{equation}
  \sup_{s\in K}|Z_E(s)-\xi(s)|
  \le C_K\|B_EP_\xi-P_\xi\|_\infty
  \int_0^1x^{\delta-1}(1-x)^{\delta-1}\dd x,
\end{equation}
where $C_K<\infty$.  This proves local uniform convergence.
\end{proof}

\section{Fermionic realization}

Let $\F_E=\Lambda^\bullet\C^E$, let $N$ be the fermion-number operator,
and define $\Str(A)=\Tr((-1)^N A)$.

\begin{theorem}\label{thm:fermionic}
The numerator \eqref{eq:numerator} satisfies the exact identity
\begin{equation}
  U_E(q)=\Str_{\F_E}\!\left(P_\xi(N/E)q^N\right).
  \label{eq:supertrace}
\end{equation}
If $u\in\R$ and $x=e^u/(1+e^u)$, then
\begin{equation}
  B_EP_\xi(x)=
  \frac{\Tr_{\F_E}\!\left(P_\xi(N/E)e^{uN}\right)}
       {\Tr_{\F_E}(e^{uN})}
  =\mathbb{E}_x\!\left[P_\xi(N/E)\right],
  \label{eq:gibbs-trace}
\end{equation}
where $N$ has the binomial distribution $\operatorname{Binomial}(E,x)$.
\end{theorem}

\begin{proof}
The $j$-particle sector of $\F_E$ has dimension $\binom{E}{j}$, and
$(-1)^N$ acts on it by $(-1)^j$.  Taking the supertrace sector by sector gives
\eqref{eq:supertrace}.  At $q=-e^u$ the parity in $q^N$ cancels the parity in
the supertrace, so
\begin{equation}
  U_E(-e^u)=\Tr_{\F_E}\!\left(P_\xi(N/E)e^{uN}\right).
\end{equation}
Since $\Tr_{\F_E}(e^{uN})=(1+e^u)^E$, division by this partition function
and the identity
$U_E(-e^u)/(1+e^u)^E=B_EP_\xi(x)$ prove \eqref{eq:gibbs-trace}.
\end{proof}

Under the standard identification of $\F_E$ with a tensor product of one-mode
spaces, the normalized operator factorizes as
\begin{equation*}
  \frac{e^{uN}}{(1+e^u)^E}=\varrho_x^{\otimes E},
  \qquad
  \varrho_x=(1-x)\lvert0\rangle\langle0\rvert
             +x\lvert1\rangle\langle1\rvert.
\end{equation*}
Thus the occupation variables are independent Bernoulli variables, with
\begin{equation*}
  \mathbb{E}_x(N/E)=x,
  \qquad
  \operatorname{Var}_x(N/E)=\frac{x(1-x)}{E}.
\end{equation*}

The reciprocity of $U_E$ also has a direct Fock-space interpretation.  Fix an
occupation basis and let $C_E$ be the particle--hole involution that sends the
basis vector indexed by a subset of $\{1,\ldots,E\}$ to the vector indexed by
its complement.  Then
\begin{equation*}
  C_ENC_E^{-1}=E-N,
  \qquad
  C_EP_\xi(N/E)C_E^{-1}=P_\xi(N/E).
\end{equation*}
Because $E$ is even, $C_E$ also preserves fermion parity.  Conjugating the
supertrace in \eqref{eq:supertrace} by $C_E$ therefore gives
\begin{equation*}
  U_E(q)=q^E U_E(q^{-1}),
\end{equation*}
so the algebraic reciprocity used by the finite transform is exactly the
particle--hole symmetry of the occupation model.

\section{Concluding remarks and open questions}

The construction supplies a concrete finite-to-continuum statement:
the polynomial transforms $Z_E$ have the same functional symmetry as $\xi$
and converge to it locally uniformly in the open critical strip.  The
mechanism differs from the usual forward use of the Rodr\'iguez-Villegas
theorem.  Here reciprocity of $U_E$ follows from the symmetry of the sampled
theta profile, but the zeros of $U_E$ are not known to lie on the unit circle.
Accordingly, local uniform convergence does not imply that every zero of
$Z_E$ lies on the critical line, and it does not control zeros that move out
of compact subsets as $E$ grows.

The trace formulas give the construction a precise finite physical
realization.  The space $\F_E=\Lambda^\bullet\C^E$ is the $2^E$-dimensional
Fock space of $E$ independent fermionic modes, and $N$ is their total
occupation number.  For real $u$, the density matrix
\begin{equation}
  \varrho_{E,u}=\frac{e^{uN}}{(1+e^u)^E}
  \label{eq:gibbs-state}
\end{equation}
is the grand-canonical Gibbs state of these modes.  Here $u$ is a
dimensionless chemical potential and
$x=e^u/(1+e^u)$ is the one-mode Fermi--Dirac occupation probability.
Equation \eqref{eq:gibbs-trace} says that the Bernstein polynomial is exactly
the thermal expectation $\Tr(\varrho_{E,u}P_\xi(N/E))$.  Thus the limit
$E\to\infty$ is simultaneously the Bernstein limit and the thermodynamic
concentration $N/E\to x$, with fluctuations of order $E^{-1/2}$.  The
reciprocity inherited from $P_\xi(x)=P_\xi(1-x)$ may be read as particle--hole
symmetry, while the numerator $U_E$ is the corresponding parity-graded trace.
The passage from $U_E(-e^u)$ to the ordinary Gibbs trace is precisely the
cancellation of this parity grading by the sign in the fugacity.  These are
standard finite-Fock-space constructions in the sense of second quantization
\cite{Berezin1966}, but here they package the theta data in an unusually rigid
polynomial form.

The Gibbs formulation points naturally toward an operator-theoretic
extension.  The theta kernel presently specifies the observable
$P_\xi(N/E)$.  It would be valuable to characterize the same observable by an
intrinsic Hamiltonian or variational principle and to determine whether the
occupation-number system embeds into a larger operator model associated with
$\xi$.  Such a construction could connect the finite thermodynamic limit
above with spectral and trace-formula structures.

A parallel geometric question is whether a deformation preserving
particle--hole symmetry also preserves the finite-to-continuum
Rodr\'iguez-Villegas limit, or whether the spaces $\Lambda^\bullet\C^E$ occur
as canonical finite compressions of the scaling or adele-class-space
frameworks of \cite{Connes1999,ConnesConsani2019}.  These questions identify
the additional structure needed to develop the exact statistical
representation into an intrinsic dynamical model.


\begin{thebibliography}{99}

\bibitem{Berezin1966}
F.~A. Berezin,
\emph{The Method of Second Quantization},
Academic Press, New York, 1966.

\bibitem{BumpChoiKurlbergVaaler2000}
D.~Bump, K.-K.~Choi, P.~Kurlberg, and J.~Vaaler,
\emph{A local Riemann hypothesis, I},
Math. Z. \textbf{233} (2000), 1--19.

\bibitem{Connes1999}
A.~Connes,
\emph{Trace formula in noncommutative geometry and the zeros of the Riemann
zeta function},
Selecta Math. (N.S.) \textbf{5} (1999), 29--106.

\bibitem{ConnesConsani2019}
A.~Connes and C.~Consani,
\emph{The Scaling Hamiltonian},
arXiv:1910.14368, 2019.

\bibitem{ConreyFarmerImamoglu2013}
J.~B. Conrey, D.~W. Farmer, and \"O.~\.{I}mamo\u{g}lu,
\emph{The nontrivial zeros of period polynomials of modular forms lie on the
unit circle},
Int. Math. Res. Not. IMRN (2013), no.~20, 4758--4771.

\bibitem{Jameson2019}
M.~Jameson,
\emph{Zeta-polynomials, Hilbert polynomials, and the Eichler--Shimura
identities},
Res. Math. Sci. \textbf{6} (2019), article 27.

\bibitem{JinMaOnoSoundararajan2016}
S.~Jin, W.~Ma, K.~Ono, and K.~Soundararajan,
\emph{The Riemann hypothesis for period polynomials of modular forms},
Proc. Natl. Acad. Sci. USA \textbf{113} (2016), no.~10, 2603--2608.

\bibitem{Lorentz1986}
G.~G. Lorentz,
\emph{Bernstein Polynomials}, 2nd ed., Chelsea Publishing Company,
New York, 1986.

\bibitem{LobrichMaThorner2017}
S.~L\"obrich, W.~Ma, and J.~Thorner,
\emph{Special values of motivic $L$-functions and zeta-polynomials for
symmetric powers of elliptic curves},
Res. Math. Sci. \textbf{4} (2017), article 26.

\bibitem{Manin2016}
Yu.~I. Manin,
\emph{Local zeta factors and geometries under $\operatorname{Spec}\,\mathbb{Z}$},
Izv. Math. \textbf{80} (2016), no.~4, 751--758.

\bibitem{OnoRolenSprung2017}
K.~Ono, L.~Rolen, and F.~Sprung,
\emph{Zeta-polynomials for modular form periods},
Adv. Math. \textbf{306} (2017), 328--343.

\bibitem{RodriguezVillegas2002}
F.~Rodr\'iguez-Villegas,
\emph{On the zeros of certain polynomials},
Proc. Amer. Math. Soc. \textbf{130} (2002), 2251--2254.

\bibitem{SchillingSongVondracek2012}
R.~L. Schilling, R.~Song, and Z.~Vondra\v cek,
\emph{Bernstein Functions: Theory and Applications}, 2nd ed.,
De Gruyter Studies in Mathematics, vol.~37, De Gruyter, Berlin, 2012.

\bibitem{Titchmarsh1986}
E.~C. Titchmarsh,
\emph{The Theory of the Riemann Zeta-Function}, 2nd ed., revised by
D.~R. Heath-Brown, Clarendon Press, Oxford, 1986.

\end{thebibliography}
\end{document}